\documentclass{article}
\usepackage[top=1in, bottom=1in, left=1in, right=1in]{geometry}
\usepackage[utf8]{inputenc}
\usepackage{amsmath,amsthm,amsfonts,amssymb}
\usepackage{authblk}
\usepackage{float}
\usepackage{url}
\usepackage{mathtools}
\usepackage{hyperref}
\usepackage[capitalise]{cleveref}
\usepackage{mathrsfs}
\usepackage{xcolor}
\usepackage{braket}
\usepackage[round]{natbib} 
\usepackage{graphicx} 
\usepackage{makecell}
\newcommand{\norm}[1]{\left\Vert {#1}\right\Vert}

\newcommand{\E}{\mathbb{E}}
\newcommand{\supp}{\text{supp}}

\newtheorem{theorem}{Theorem}[section]

\newtheorem*{theorem*}{Theorem}
\newtheorem{lemma}[theorem]{Lemma}
\newtheorem*{lemma*}{Lemma}
\newtheorem{proposition}[theorem]{Proposition}
\newtheorem*{prop*}{Proposition}

\newtheorem{remark}[theorem]{Remark}

\newtheorem{definition}[theorem]{Definition}
\newtheorem{fact}[theorem]{Fact}

\usepackage{comment}

\makeatletter
\let\c@equation\c@theorem
\makeatother

\newcommand{\N}[0]{\mathbb{N}}

\newcommand{\R}[0]{\mathbb{R}}

\newcommand{\poly}{\operatorname{poly}}

\newcommand{\rank}{\operatorname{rank}}

\newcommand{\Tr}[0]{\operatorname{Tr}}

\newcommand{\pull}[9]{
#1\ar@/_/[ddr]_{#2} \ar@{.>}[rd]^{#3} \ar@/^/[rrd]^{#4} & &\\
& #5\ar[r]^{#6}\ar[d]^{#8} &#7\ar[d]^{#9} \\}

\newcommand{\cmp}[9]{
\xymatrix{
#1 \ar[r]^{#4}{#5} \ar@/_2pc/[rr]^{#8}_{#9} & #2 \ar[r]^{#6}_{#7} & #3
}
}

\newcommand{\ha}[1]{\ar@{^(->}[#1]}
\newcommand{\ls}[1]{\ar@{-}[#1]}
\newcommand{\sj}[1]{\ar@{->>}[#1]}
\newcommand{\aq}[1]{\ar@{=}[#1]}
\newcommand{\acir}[1]{\ar@{}[#1]|-{\textstyle{\circlearrowright}}}
\newcommand{\acil}[1]{\ar@{}[#1]|-{\textstyle{\circlearrowleft}}}
\newcommand{\ard}[1]{\ar@{.>}[#1]}
\newcommand{\mt}[1]{\ar@{|->}[#1]}
\newcommand{\inm}[1]{\ar@{}[#1]|-{\in}}
\newcommand{\inr}{\ar@{}[d]|-{\rotatebox[origin=c]{-90}{$\in$}}}
\newcommand{\inl}{\ar@{}[u]|-{\rotatebox[origin=c]{90}{$\in$}}}

\newcommand{\beq}[1]{\begin{equation}\llabel{#1}}
\newcommand{\eeq}[0]{\end{equation}}
\newcommand{\bal}[0]{\begin{align*}}
\newcommand{\eal}[0]{\end{align*}}
\newcommand{\ban}[0]{\begin{align}}
\newcommand{\ean}[0]{\end{align}}

\newcommand{\fixme}[1]{{\color{red}#1}}
\newcommand{\llabel}[1]{\label{#1}\text{\fixme{\tiny#1}}}

\newcommand{\arxiv}[1]{\url{http://www.arxiv.org/abs/#1}}

\allowdisplaybreaks[2]

\DeclareFontFamily{U}{wncy}{}
    \DeclareFontShape{U}{wncy}{m}{n}{<->wncyr10}{}
    \DeclareSymbolFont{mcy}{U}{wncy}{m}{n}
    \DeclareMathSymbol{\Sh}{\mathord}{mcy}{"58} 

\usepackage{makecell}
\newcommand{\OPT}{\text{OPT}}
\newcommand*{\B}{\mathcal{B}}
\newcommand*{\matroid}{\mathcal{M}}
\newcommand*{\Sd}{\mathbb{S}_d^+}
\newcommand{\abs}[1]{|#1|}
\usepackage{appendix}
\title{Composable Coresets for Constrained Determinant Maximization and Beyond}

\author[1]{Sepideh Mahabadi}
\author[2,3]{Thuy-Duong Vuong}
\affil[1]{Microsoft Research,\texttt{smahabadi@microsoft.com}}
\affil[2]{Computer Science Division, University of California, Berkeley, CA 94720, USA, and Miller Institute for Basic Science Research. \texttt{tdvuong@berkeley.edu}.}
\affil[3]{Department of Computer Science and Engineering, University of California, San Diego, CA 92093, USA. \texttt{thvuong@ucsd.edu} (preferred email)}

\begin{document}


\maketitle

\begin{abstract}
We study algorithms for construction of {\em composable coresets} for the task of {\em Determinant Maximization}
under {\em partition constraint}. Given a point set $V\subset \R^d$ that is partitioned into $s$ groups $V_1,\cdots, V_s$, and integers $k_1,...,k_s$, where $k=\sum_i k_i$,
the goal is to pick $k_i$ points from group $V_i$ such that the overall determinant of the picked $k$ points is maximized.
Determinant Maximization and its constrained variants have gained a lot of interest for modeling diversity, and have found applications in the context of data summarization.


When the cardinality $k$ of the selected set is greater than the dimension $d$, we show a peeling algorithm that gives us a composable coreset of size $kd$ with a provably optimal 
approximation factor of $d^{O(d)}.$ When $k\leq d$, we show a simple coreset construction with optimal size and approximation factor. As a further application of our technique, we get a composable coreset for determinant maximization under the more general laminar matroid constraints, and a composable coreset for unconstrained determinant maximization in a previously unresolved regime.

Our results generalize to
all strongly Rayleigh distributions and to several other experimental design problems.
As an application, we improve the runtime of the practical local-search based algorithm of [Anari-Vuong--COLT'22] for determinantal maximization under partition constraint from $O(n^{2^s}k^{2^s})$ to $O(n k^{2^s})$, making it only linear on the number of points $n$.

\end{abstract}


\section{Introduction}
Determinant maximization is a fundamental optimization problem that arises in various domains such as data summarization, experimental design, computational geometry, and machine learning. At its core, the problem involves selecting a subset of items, typically vectors, such that the selected set is as diverse or informative as possible. A common way to quantify this diversity is via the determinant of a submatrix derived from the selected vectors.

Before formally defining the problem, let us start by recalling the notation $\det_k$, a generalization of the standard determinant. 
\begin{definition}[$\det_k$]
Given a $d\times d$ matrix $M\in \R^{d\times d}$ and $k\leq d$, $\det_k(M)$ denotes the sum of determinants of all $k\times k$ principal submatrices of $A$:
\[\det_k (M)= \sum_{|S|=k} \det(M_{S,S}),\]
where $M_{S,S}$ is the principal sub-matrix indexed by $S\subseteq [d]$.
\end{definition}

Given a collection of $n$ vectors $V=\{v_1,\cdots,v_n\}$ in $\R^d$, and a target subset size $k$ (not necessarily smaller than $d$), the \emph{determinant maximization} problem seeks a subset $S\subseteq [n]$, $|S|=k$, that maximizes the quantity 
\[\phi(S) := \det_{\min\{k,d\}}(L(S))
\]
where $L(S):= A_S A_S^\intercal  =\sum_{i\in S} v_i v_i^{\top}\in\R^{d\times d}$ and $A_S$ is the $d\times k$ matrix whose columns are the vectors $v_i$ for $i\in S$.

Geometrically, if $k\leq d$, the objective function is equal to the volume squared of the parallelepiped spanned by the selected vectors. In this setting, the best approximation guarantee is $e^k$ \cite{nikolov2015randomized}, which is essentially tight \cite{civril2013exponential} unless P = NP. 

On the other hand, if $k>d$, 
the objective can be rewritten as $ \det (A_S^\intercal A_S )$, and problem is known as the D-optimal design problem. The best known approximation algorithm in this regime achieves a factor of $\min\{e^k , (\frac{k}{k-d})^d\}$, which is always at most $\leq e^{O(d)}$
and becomes a constant when for example $k\geq d^2$ \cite{madan2019combinatorial}.

Due to its connection with subset diversity, determinant maximization and its variants have been widely used in modern data analysis, particularly for summarization tasks where one aims to select a compact, informative, and representative subset from large-scale data, and thus studied extensively over the last decade \cite{mirzasoleiman2017streaming,gong2014diverse, kulesza2012determinantal, chao2015large, kulesza2011learning, yao2016tweet, lee2016individualness}.

\textbf{Determinant Maximization under partition and matroid constraints.} 
Diversity maximization problems, including determinant maximization, have been studied extensively under partition and more generally under matroid constraints \cite{Madan2020MaximizingDU, nikolov2016maximizing, abbassi2013diversity, moumoulidou2020diverse, addanki2022improved, mahabadi2023core}. 
These constraints are important in real-world applications where certain fairness or grouping criteria must be respected.
In the simpler case of a partition constraint, the data set $V$ is partitioned into $s$ groups $V_1,\cdots, V_s$ and we are provided with $s$ numbers $k_1,\cdots,k_s$, and the goal is to pick $k_i$ points from each group $i$ such that the overall determinant (or more generally diversity) of the chosen $k=\sum_i k_i$ points is maximized. 
This setting allows for fine-grained control over the contribution of each group in the selected summary: for example, limiting the number of movies from each genre in a recommendation system, and has further applications in the context of fair and balanced data summarization (see e.g. \cite{mahabadi2023core}).
More generally, given a matroid $([n],\mathcal{I})$ of rank $k$, the problem of finding a basis of the matroid that maximizes the determinant admits a $\min\{k^{O(k)},d^{O(d)}\}$ approximation, and it improves to $\min\set{e^{O(k)}, d^{O(d)}}$ for the {\em estimation} problem where the goal is to only estimate the value of the optimal solution \cite{Madan2020MaximizingDU, nikolov2016maximizing,BLPST22,brown2022efficient}.

\textbf{Composable Coresets.} 
As one of the main applications of determinant maximization is in data summarization, the problem has been considered extensively in massive data computation models \cite{mirzasoleiman2017streaming, wei2014fast, pan2014parallel, mirzasoleiman2013distributed, mirzasoleiman2015distributed, mirrokni2015randomized, barbosa2015power}. In this work, we focus on designing \emph{composable coreset} for determinant maximization. 
A {\em Coreset} is a {\em small subset} of the data that is sufficient for computing an approximate solution to a pre-specified optimization problem on the whole dataset \cite{agarwal2005geometric}. 
More specifically, we present a summarization algorithm $\mathcal{A}$ that, given a data set $V$, produces a subset of $V$.
Moreover, we want our coresets to be composable \cite{IMMM-ccdcm-14}: that is if we have multiple datasets $V^{(1)},\cdots, V^{(m)}$ (note that in the context of, e.g., a partition constraint, each data set $V^{(i)}$ is itself partitioned into groups $V^{(j)}_1,\cdots,V^{(j)}_s$),
then the union of coresets $\mathcal{A}(V^{(1)})\cup \cdots \cup \mathcal{A}(V^{(m)})$ should be sufficient for computing an approximate solution for the union of the datasets $V^{(1)}\cup\cdots\cup V^{(m)}$ (see \cref{sec:coreset} for a formal definition).

As shown in \cite{IMMM-ccdcm-14}, having a composable coreset for an optimization task automatically yields a solution for the same task in several massive data models including distributed/parallel and streaming models. For example, in a distributed setting where the whole data is partitioned over multiple machines, each machine can compute a coreset for its own data, and only send this small summary to a single aggregator. The aggregator then processes the union of the summaries and outputs the solution. Due to their applications, several works have focused on designing composable coresets for determinant maximization, and more broadly, diversity maximization, over the past decade \cite{IMMM-ccdcm-14, indyk2020composable, mahabadi2019composable, mirrokni2015randomized, moumoulidou2020diverse, ceccarello2018fast, ceccarello2020general, zadeh2017scalable, mahabadi2023improved, mahabadi2023core}.

\textbf{Prior work.} Composable coresets have been designed for the unconstrained determinant maximization problem. More precisely, for $k\leq d$, one can get a $k^{O(k)}$- approximate coreset of size $O(k),$ which is also known to be tight \cite{indyk2020composable, mahabadi2019composable} (see the first row of \cref{table:results}). 
Furthermore, for $k\geq d$, if the solution is allowed to pick vectors from $V$ {\em with repetition} (which we refer to as the ``with-repetition" setting), then \cite{mahabadi2019composable} gives a coreset of size $\tilde O(d)$ with approximation factor $\tilde O(d)^d$ (Second row of \cref{table:results}). However, the case where the solution is required to pick {\em distinct} points from $V$ (also known as the ``without-repetition" setting) remains open:  when $k\gg d$, \cite{mahabadi2019composable}'s size-$\tilde{O}(d)$ coreset clearly does not contain enough points to construct a solution that consists of $k$ distinct points.
Determinant maximization in the without-repetition setting is generally harder and has received significantly more interest than the with-repetition setting \cite[see][]{madan2019combinatorial,lau2021local,BLPST22}, as its solution often provides a much better summary of the original dataset \footnote{Consider the following illustrative example. Suppose $d =2$ and the data set $V$ consists of $v_1\equiv M e_1, v_2 \equiv M e_2, v_3, \dots, v_n$ where $e_1,e_2$ are standard basis vectors and $v_3, \cdots, v_n$ are arbitrary vectors in $\R^2.$ For large enough $M,$ the coreset $C$ of $V$ for $k\geq 2$ in the "with repetition" setting will only contain the 2 vectors $v_1 $ and $v_2,$ since the set consisting of $k/2$ copies of $v_1$ and of $v_2$ has the biggest possible determinant. But we cannot say that this set is a good summary of the data set since we ignore the information provided by $v_3, \dots, v_n.$ On the other hand, a coreset for $k \geq 2$ in the "without repetition" setting would necessarily consider the remaining vectors $v_3, \dots, v_n,$ and thus is a better representation of the data set. 
}.

\subsection{Our Results}
In this work, we establish the following contributions.

\textbf{Algorithms.} A summary of our coreset construction algorithms is provided in \cref{table:results}.

\begin{table}[!htb]
    \centering
    \begin{tabular}{|c|c|c|c|}
    \hline
    & size & approximation & constraint\\
    \hline
    $k\leq d$ & $k$ & $O(k)^{2k}$ & \makecell{cardinality\\ \cite{indyk2020composable, mahabadi2019composable}}   \\
    \hline
    $k\geq d$ & $\tilde O(d)$ & $\tilde O(d)^{2d}$ & \makecell{cardinality 
    \\(with-repetition)\\ \cite{indyk2020composable}}  \\
    \hline
    $k\geq d$ & $kd$ & $d^{2d}$ & \makecell{cardinality\\ (without-repetition)\\ {\color{red}This work} }  \\
    \hline
    $k \leq d$ & $sk$ & $k^{2k}$ & \makecell{partition \\{\color{red}This work}}  \\
    \hline
    $k\geq d$ & $kd$ & $d^{2d}$ &\makecell{partition\\{\color{red}This work}} \\
    \hline
    $k \leq d$ & $k^{2k}$ & $k^{2k}$ & \makecell{laminar\\{\color{red}This work}}\\
    \hline
    $k\geq d$ & $(kd)^{k}$ & $d^{2d}$ &\makecell{laminar\\{\color{red}This work}} \\
    \hline
    \end{tabular}
    \caption{\small
    Our upper bound results on composable coresets for determinant maximization. Here $s$ is the number of groups in the partition constraints. Note that the third row follows from our results on partition constraint but we spell it out to compare to the previous result.}\label{table:results}
\end{table}
\begin{itemize}
\item In \cref{sec:improve exchange}, we construct coresets for unconstrained determinant maximization in the without-repetition setting (third row in \cref{table:results}), previously left open in prior work.

\item 
 In \cref{sec:partition-laminar}, we develop efficient composable coresets for determinant maximization under {\em partition} and {\em laminar} matroid constraints, as shown in rows 4–7 of  \cref{table:results}, and verified in \cref{thm:main}. Our results extends to {\em Strongly Rayleigh} distributions (see \cref{thm:main rayleigh laminar}), and are obtained via an improved exchange inequality for determinant when $k > d.$
\item
In \cref{sec:design} (\cref{thm:design}),we demonstrate an application of our results to design composable coresets for a broader class of experimental design problems in the without-repetition setting, for all regular objective functions. This complements the result of 
\cite{indyk2020composable} for experimental design in the with-repetition setting.
\end{itemize}

\paragraph{Lower bounds.} We complement our results with the following lower bounds shown in \cref{sec:lower-bound}. 
\begin{itemize}
\item In \cref{lem:size lb in low dimension}, we show that for $k\leq d$, any composable coreset for determinant maximization under partition constraint with a finite approximation factor must have size at least $\Omega(sk)$. This shows that our construction (fourth row in \cref{table:results}) is essentially tight, since our approximation factor matches that of  \cite{indyk2020composable} for the unconstrained version of the problem, which is known to be tight.

\item 
In \cref{lem:size lb in high dimension}, we show that for $k\geq d$, any composable coreset for the problem under partition constraint
with a finite approximation factor must have size  at least $ k + d(d-1)$. This partially complements our result in the fifth row of \cref{table:results}, and shows that for $k=O(d)$, our coreset size cannot be improved.   

\item In \cref{thm:lowerbound}, we prove that for $ d\leq k \leq \poly(d)$ and coreset of polynomial size in $k$, the approximation factor of $d^{O(d)}$ is essentially the best possible for the unconstrained determinant maximization problem in the without-repetition setting. 
This matches the approximation factor of our construction in the third row of \cref{table:results}.
\end{itemize}

\paragraph{Application.} Our coreset can be constructed in essentially linear time in $n,$ the number of data points. Hence, by first constructing a coreset then applying any standard determinant maximization algorithm on the coreset, we obtain an algorithm for the determinant maximization problem under partition constraint that runs in time $O(n \poly(k))$ (see \cref{lem:linear time algo partition matroid}). Thus, if we use the  multi-step local search algorithm by \cite{anari2021sampling} for determinant maximization under partition constraint, then we obtain a practical local-search-based algorithm whose runtime improves from $O(n^{2^s}k^{2^s})$ in \cite{anari2021sampling} to $O(n k^{2^s})$ where $s$ is the number of parts in the partition. 
\subsection{Overview of the Techniques}


Let us give a brief overview of our approach.
Consider a set of vectors $V = \set{v_1, \cdots, v_n}\subseteq \R^d.$ 
As mentioned earlier, when $k\leq d$, the objective function for any subset $T \in \binom{[n]}{k}$ defined as $\phi(T)=\det_{k} (\sum_{i\in T} v_i v_i^\intercal)$ corresponds to the square of the volume spanned by the vectors in $V_T :=\set{v_i | i \in T}$.
\citet{mahabadi2019composable} showed that in this setting, any local maximum $U\subseteq V$ of size $k$ with respect to $\det_k(\cdot)$ approximately preserves the  {\em $k$-directional height} of the set $V$. More precisely, for any set $V_S\subset \R^d$ of $k$ vectors and any $v\in V_S\cap V,$ one can replace $v$ with some $u$ in $U$ so that the distance $d(u, \mathcal{H})$ from $u$ to the $(k-1)$-dimensional subspace $\mathcal{H}$ spanned by $ V_S\setminus \set*{v}$ is at least $\frac{1}{k}\cdot d(v, \mathcal{H})$. Thus
\[ k^2 \det_k(u u^\intercal + \sum_{w \in V_S\setminus \set*{v}} w w^\intercal) \geq  \det_k(\sum_{w \in V_S} w w^\intercal).\] 
Thus, all elements of $V_S$ can be successively replaced by elements of $U$ while only incurring a factor $k^{O(k)}$ increase in the objective function $\det_k (\cdot)$ and
 thus $U$ is a $k^{O(k)}$ composable coreset w.r.t $ \det(\cdot).$ 
 
 \noindent\textbf{Generalization of directional height.} In this work, we {\em extend} the notion of directional height to the regime where $k \geq d$. We show that for any local maximum $U$ of size $d$ with respect to $\det(\cdot)$, the following holds: For any index set $S\subseteq [n]$ of size $k,$ and letting $V_S = \set{v_i:i\in S}$ and $v\in V_S\cap V$, there exists $u\in U$ s.t.
\[\det(d^2 u u^\intercal + \sum_{w\in V_S\setminus \set*{v} } w  w^\intercal ) \geq \det(\sum_{w\in V_S} w w^\intercal)\footnote{This is precisely $\phi(T) = \det_{\min\{k,d\}} (\cdot) $ when $k\geq d.$} .\]

This already implies that $U$ forms a $d^{O(d)}$-composable coreset for determinant maximization in the with-repetition setting, i.e., when the selected subset is allowed to contain duplicate vectors.

\noindent\textbf{The without-repetition setting.} The without-repetition case is more delicate, as we must ensure that $ (V_S\setminus \set*{v}) \cup \set*{u}$ is a proper subset, i.e., $u\not\in  V_S\setminus \set*{v}.$ To handle this, we apply the idea of peeling coresets, previously used for constructing robust coresets that tolerate outliers \cite{agarwal2008robust, abbar2013diverse}. Our construction repeatedly peels away local optimum solutions from the input set, and takes the union of all the peeled local optimums to be the final coreset. By the pigeon hole principle, for any set $V_S$ not fully contained in the final coreset, there exists at least one peeled-away local optimum that is disjoint from $V_S$. Consequently, we can replace an element of $V_S$ by an element inside this local optimum without creating a multiset. 

\noindent\textbf{Partition and laminar constraints.} For determinant maximization under partition constraint,  our coreset construction is simple and intuitive when $k\leq d$: we take the union of the coresets for each partition part. 
When $k\geq d,$ we construct a coreset of size $kd$ with approximation factor $d^{O(d)}$ by taking the union of the peeling coresets for each part of the partition.
For the laminar matroid constraint case, we apply the peeling coreset idea to ensure that for any subset $V_S$ satisfying the laminar constraint, there exists one peeled-away subset $U$ s.t. replacing an element of $V_S$ by an element of $U$ will not violate the laminar constraint.


\noindent\textbf{Strongly Rayleigh distributions.} The fixed-size determinantal point processes (DPP), a distribution over subsets $\binom{[n]}{k}$ defined by $\P{S}\propto \det(S),$ belongs to the class of {\em Strongly Rayleigh} distributions (see \cref{subsec:strongly Rayleigh} for details). 
All of our results readily extend to maximum a posteriori problems, i.e., find $\arg \max_S \mu(S)$ for any strongly Rayleigh objective functions $\mu(\cdot)$. This is because our proof relies only on an {\em exchange inequality} which is satisfied by all strongly Rayleigh distributions. Exchange inequalities offer a unifying framework for analyzing our coreset constructions across all settings. To the best of our knowledge, this is the first work to leverage exchange inequalities in the context of coreset construction.

\noindent\textbf{Experimental design.} Our construction also applies to experimental design problems with respect to other, not necessarily strongly Rayleigh, objective functions, such as matrix traces (A-design) or condition number (E-design). 
By replacing the base-level building blocks in our construction, i.e., the local optimum w.r.t $\det(\cdot)$, with spectral spanners \cite{indyk2020composable}, we can guarantee that the union of the coresets contains a feasible fractional solution as a combination of input vectors that achieves a good value.
However, the algorithm to round the fractional solution to an integral solution under matroid constraint only exists in limited cases of objective function other than the determinant. 


\noindent\textbf{Lower bounds.}
For lower bound on the size of composable coreset for determinant maximization under partition constraint, when $k\leq d$, we show that any coreset that achieves a finite approximation factor must include at least $k$ vectors from each part of the partition. 
When $k\geq d$, we show an analogous result: any such coreset must include at least $d$ vectors from at least $d$ d parts of the partition. Finally, to prove that the approximation factor of $d^{O(d)}$ is the best possible when $\poly(d) \geq k \geq d,$ we use a similar construction as the one used in \cite{indyk2020composable}'s lower bound for the approximation factor when $k\leq d.$  
\section{Preliminaries}
Let $[n]$ denote the set $ \set*{1,\cdots, n}.$
For a set $U,$ we use $\binom{U}{k}$ to denote the family of all size-$k$ subsets of $U.$
For a set $V$ of vectors $\{v_1, \cdots, v_n\}$ and $S\subseteq [n],$ we use $V_S$ to denote the set $\set{v_i| i \in S}.$ For sets $ U,W,$ we use $U+W$ and $U-W$ to denote $U\cup W$ (union) and  $U\setminus W$ (set-exclusion) respectively. For singleton subsets, we abuse notation and write $U-e$ ($U+e$ resp.) for $ U- \set*{e}$ ($U+\set*{e}$ resp.).

For a matrix $M\in \R^{n\times n}$ and $S\subseteq [n]$,  we use $M_S$ to denote the principal submatrix of $M$ whose rows and columns are indexed by $S.$ 
We use $\Sd$ to denote the set of all symmetric positive semi-definite matrices in $ \R^{d\times d}.$

\begin{definition}[Local optima] \label{def:local optima}
For a function $\mu: \binom{[n]}{k}\to \R_{\geq 0}$ and $\zeta \geq 1$, we say $ U$ is an $\zeta$-approximate local optima of $\mu$ iff $\zeta \mu(U) \geq  \max_{e\in U, f \in [n] \setminus U}\mu (U-e+f).$

When $ \zeta = 1,$ we simply refer to $U$ as a local optima.
\end{definition}



\subsection{Matroids}
We say a family of sets $ \mathcal{B} \subseteq \binom{[n]}{k}$ is the family of bases of a matroid if $\mathcal{B}$ satisfies the basis exchange axiom: for any two bases $ B_1, B_2 \in \mathcal{B}$ and $x \in B_1 \setminus B_2,$ there exists $y \in B_2 \setminus B_1$ such that $ B_1 -x + y \in \mathcal{B}.$  We call $k$ the rank of the matroid, and $[n]$ the ground set of the matroid.
We let the family of independent sets of the matroid be $\mathcal{I} = \set*{I \in 2^{[n]} | \exists B \in \B: I \subseteq B }.$

The family $\binom{[n]}{k}$ of all size-$k$ subsets of $[n]$ forms the set of bases of the \emph{uniform matroid} of rank $k$ over $[n].$
We define two simple classes of matroids that are widely used in applications.
\begin{definition}[Partition matroid]
Given a partition of $[n]$ into $P_1, \cdots, P_s$ and integers $k_1, \cdots, k_s,$ the associated partition matroid is defined by: a set $ I \subseteq [n]$ is independent iff $ \abs{I \cap P_i} \leq k_i$, for $\forall i\in [s].$ The rank of the matroid is $k: = \sum_{i=1}^s  k_i.$ 
\end{definition}
\begin{definition}[Laminar matroid]
A family $\mathcal{F}$ of subsets of $[n]$ is laminar iff for any $ F_1, F_2 \in \mathcal{F}$ either $F_1 $ and $F_2$ are disjoint or $F_1$ contains $F_2$ or $F_2$ contains $F_1.$  Given a laminar family $\mathcal{F}$ and integers $k_F$ for each $F\in \mathcal{F} ,$ the associated laminar matroid is defined by: a set $I \subseteq [n]$ is independent iff $\abs{I \cap F} \leq k_F$ for $\forall F \in \mathcal{F}.$ The maximal independent sets have the same cardinality $k,$ and they form the bases of the laminar matroid.

We assume that $ k_F > 0$, for $\forall F \in \mathcal{F},$ otherwise we can remove the set $F$ from $\mathcal{F}$ and all elements in $F$ from the ground set.
For two sets $F_1, F_2$ in $\mathcal{F}$ s.t. $ F_1 \subseteq F_2,$ we can assume $k_{F_2} > k_{F_1},$ otherwise the constraint on $I \cap F_1$ is redundant and $F_1$ can be removed from the laminar family. We call such a laminar family \emph{non-redundant}.
\end{definition}

\subsection{Determinant Maximization and Experimental Design Problems}
Given vectors $v_1, \cdots, v_n\in \R^d$ and a matroid $\mathcal{M}=([n], I),$ determinantal point processes (DPP) under matroid constraint samples a basis $S\subseteq [n]$ of $\mathcal{M}$ such that 
\[\P{S} \sim \det_{\min \{k,d\}} (\sum_{i\in S}  v_i v_i^\intercal).\]
This distribution favors diversity, since sets of vectors that are more linearly independent (i.e., different from each other) are assigned higher probabilities. The fundamental optimization problem associated with DPPs, and probabilistic model in general, is to find a "most diverse" subset by computing $\arg\max_{S \text{ is a basis of } \mathcal{M}} \P{S}$ i.e. solving the maximum a posteriori (MAP) inference problem. 

When the matroid is the uniform matroid, we simply refer to the problem as the {\em determinant maximization} problem.

When $k \leq d$, $\P{S}$ is proportional to the squared volume of the parallelepiped spanned by the elements of $S$. T  hus MAP-inference for DPPs is also known as the volume maximization problem. 


Determinant maximization is also known as the $D$-design problem, since the objective function is the (D)eterminant. Other objective functions have also been studied, for example, matrix traces (A-desing) or condition number (E-design). We discuss these different objective functions in more details in \cref{sec:design}. 

The setting where $S$ is allowed to be a multiset has also been studied. 
This is known as the experimental design problem \emph{with-repetition} \cite{roundingExperimentalDesign,madan2019combinatorial}, as opposed to the \emph{without-repetition} setting where $S$ needs to be a proper subset. The with-repetition setting is generally easier: it can be reduced to the without-repetition setting by duplicating each vector $k$ times.

\subsection{Composable Coresets}\label{sec:coreset}
In the context of the optimization problem on $\mu: \binom{[n]}{k}\to \R_{\geq 0}$, a function $c$ that maps any set $V\subseteq [n]$
to one of its subsets is called an $\alpha$-composable coreset (\cite{IMMM-ccdcm-14}) if it
satisfies the following condition: given any 
collection of $m$ sets $V^{(1)}, \cdots , V^{(m)} \subseteq [n]$
\begin{align*}
  &\alpha \cdot \max\set*{\mu(S) | S \subseteq \bigcup_{i=1}^m c(V^{(i)})} 
\geq \max \set*{\mu(S) | S \subseteq \bigcup_{i=1}^m V^{(i)} }  
\end{align*}
If we are further given a matroid $\mathcal{M}=([n], I)$ to satisfy, we additionally require $S$ to be a basis of $\mathcal{M}$ in both sides of the above inequality.
Finally, we say that $c$ is a coreset of {\em size} $ t$ if $|c(V)| \leq t$ for all sets $V$. 
Composable coresets are very versatile; once a composable coreset is designed for a task, it automatically imples
efficient streaming and distributed algorithms for the same task.

\subsection{Directional Height}
For this subsection let $k\leq d$. 
\begin{definition}[Directional height and $k$-directional height \cite{mahabadi2019composable}]
For a set $V\subseteq \R^d$ of vectors and a unit vector $x,$ the directional height of $V$ w.r.t $x$ is $h(V,x) = \max_{v\in V} \abs{\langle v, x \rangle}.$

The $k$-directional height of $V$ w.r.t a $(k-1)$-dimensional subspace $H$ is $d(V, H)  = \max_{v\in V, x\in H^{\top}} \abs{\langle v, x \rangle}$ where $H^{\top}$ is the $(d-k+1)$-dimensional subspace perpendicular to $H.$
\end{definition}

\begin{theorem}[Coreset for $k$-directional height \cite{mahabadi2019composable}]\label{thm:preserve k directional height}
Let $k \leq d$ and $V\subseteq\R^d$. Then any size $k$ local optimum $U$ w.r.t $\det(\cdot)$ inside $V$ approximately preserves the $k$-directional height. That is, for any $(k-1)$-dimensional subspace $H$
\[d(U, H) \geq \frac{1}{k} d(V, H).\]
where for a point set $P$, we define $d(P,H)=\max_{p\in P} d(p,H)$.
\end{theorem}




\subsection{Strongly Rayleigh Distribution and Exchange Inequalities} \label{subsec:strongly Rayleigh}
Let $\nu: \binom{[n]}{\ell} \to \R_{\geq 0},$ be a distribution over size-$\ell$ subsets of $[n]$. Its generating polynomial is defined as
\[g_{\mu}(z_1, \cdots, z_n) = \sum_{S \in \binom{[n]}{\ell}} \nu(S) \prod_{i\in S} z_i.\]

\begin{definition}[Strongly Rayleigh]
A distribution $\nu: \binom{[n]}{\ell} \to \R_{\geq 0}$ is {\em strongly Rayleigh} (or {\em real-stable}) if its generating polynomial $g_{\nu}$ has no roots in the upper-half of the complex plane. That is,
$g_{\nu}(z_1, \cdots, z_n) \neq 0$ whenever $ \Im(z_i) > 0$ for $\forall i. $
\end{definition}


Strongly Rayleigh distributions satisfy the following {\em exchange inequality}, which implies that for any local optimum subset $U$ (w.r.t a 
strongly Rayleigh distribution $\nu,$), and any set $W\in \supp(\nu),$ we can replace an element of $W$ for an element of $U$ while approximately preserving the value of $\nu(\cdot).$ This property will later be used to show that $U$ can serve as a coreset by successively replacing all elements of $W$ with elements of $U$ without significantly losing the objective value.
\begin{lemma}{Exchange inequality \cite[Lemma 26]{ALOVV21} } \label{lem:exchange inequality}
Let $\nu: \binom{[n]}{\ell} \to \R_{\geq 0}$ be a strongly Rayleigh distribution.
Let $V \subseteq [n]$ be an arbitrary subset. For $ \zeta \geq 1$, let $U\subseteq V$ be a $\zeta$-local optimum w.r.t. $\nu$, and assume $\nu(U) \neq 0.$
Then for any $W \in  \binom{[n]}{\ell},and $ $e \in W \setminus U$
\[\nu(W) \nu(U) \leq \ell\cdot \sum_{j\in U\setminus W} \nu(W - e+j) \nu(U+e-j)\]
In particular, if $ e \in V$, then by approximate local optimality of $U$ within $V$, we have
\[\nu(W) \leq \zeta \ell \cdot \sum_{j\in U} \nu(W - e+j) \leq(\zeta\ell)^2 \cdot \max_{j\in U} \nu(W-e+j)\]
where we implicitly understand that if $j \in W-e$ then $W-e+j$ is not a proper set, and $\nu(W-e+j) = 0.$
\end{lemma}
In the context of determinant maximization with given input vectors $v_1, \cdots, v_n \in \R^d$, for any $k\leq d$, $\nu(S) = \det_k(\sum_{i\in S} v_i v_i^\intercal)$ defines a strongly Rayleigh distribution over subsets $S\in \binom{[n]}{k}$ \cite{BBL09,anari2016monte}.

\section{Unconstrained Case: the Peeling Coreset} \label{sec:improve exchange}
As mentioned in the overview of the techniques, if $U$ is a composable coreset for $V$ with respect to the function $\mu(\cdot)$, 
  then for any set $S,$ we can replace any element in $S\cap V$ with an element of $U$ while not significantly reducing $\mu(\cdot)$. We formalize this intuition with the following definition. 
\begin{definition}[Value-preserving set]\label{def:value-preserving}
Given $\tilde{\mu}: [n]^k\to \R_{\geq 0}$ and $V \subseteq [n],$ we say $ U \subseteq V$ is \emph{value-preserving} with respect to $\tilde{\mu}$ if for any $S\in  \binom{[n]}{k} $ and $e \in S \cap V,$ there exists $f \in U \setminus ( S-e)$ s.t. $\tilde{\mu}(S ) \leq \tilde{\mu}(S - e+ f).$
\end{definition}
The following lemma shows the relationship between value-preserving sets and composable coresets.
\begin{lemma} \label{lem:value preserving to coreset}
Suppose functions $\mu ,\tilde{\mu}: [n]^k\to \R_{\geq 0}$ satisfy that $\mu(S)\leq \tilde{\mu}(S) \leq \alpha \mu(S)$ for all $S,$ and let $c$ be a coreset map $c$ such that $U:=c(V)\subseteq V$ is \emph{value-preserving} with respect to $\tilde{\mu}$. Then $c$ gives an $\alpha$-composable coreset with respect to $\mu.$
\end{lemma}
\begin{proof}
Consider a collection of datasets $V^{(1)}, \cdots, V^{(m)} ,$ let $U_i: = c(V^{(i)})$ for $\forall i\in [m],$ and let $S$ be an arbitrary size-$k$ subset of $\bigcup_{i=1}^m V^{(i)}.$ Since each $U_i$ is value-preserving w.r.t. $\tilde{\mu},$ for any $ e\in S\cap V^{(i)},$ we can replace $e$ with $ f\in U_i \setminus (S-e)$ while keeping $\tilde{\mu}(\cdot)$ non-decreasing. Thus, we successively replace  $E:=S\setminus \bigcup_{i=1}^m U_i$ with  $L \subseteq \bigcup_{i=1}^m U_i$ while ensuring that
\[\tilde{\mu}(S) \leq \tilde{\mu}(S- E +L)\]
Moreover, $S-E +L \in \bigcup_{i=1}^m U_i$ and \[\mu(S) \leq \tilde{\mu}(S) \leq \tilde{\mu}(S- E +L) \leq \alpha \mu(S-E+L). \]
By choosing $S$ such that $\mu(S) = \max\set*{\mu(S') | S' \subseteq \bigcup_{i=1}^m V^{(i)}, |S'|=k}$ we get the desired conclusion.
\end{proof}

We now show that an (approximate) local optima with respect to $\det(\cdot)$ is value preserving for suitably chosen functions. When $ k \leq d,$ \cite{mahabadi2019composable} shows that any size-$k$ local optimum $U$ with respect to $\det (\cdot )$ approximately preserves $k$-directional height (see \cref{thm:preserve k directional height}), and hence is a value-preserving set with respect to $\tilde{\mu},$ where $\tilde{\mu}$ is defined by 
\[\tilde{\mu}(S) = \det_k(\sum_{i\in S} k^{2\times \mathbf{1}[i \in U]} v_i v_i^\intercal )  \]
This is easy to see since for $\abs{S}=k\leq d,$ as $\mu(S)= \det_k(\sum_{v\in S} v v^\intercal)$ is precisely the square of the volume spanned by vectors in $S$, i.e., $\mu(S) = \text{Vol}^2(\set*{v_i | i\in S})$ and $ \tilde{\mu}(S) =k^{2\abs{U\cap S} } \text{Vol}^2(\set*{v_i|i\in S}). $

Below, we show that for $k \geq d,$ a size-$d$ local optimum $U$ is value-preserving with respect to $\tilde{\mu}$ defined by
\begin{equation} \label{eq:modified determinant}
\begin{split}
    \tilde{\mu}(S) = \det(\sum_{i\in S} d^{2 \times \mathbf{1}[i \in U]} v_i v_i^\intercal ) 
    =\sum_{W \in \binom{S}{d}} d^{2\abs{W \cap U} } \det(\sum_{i\in W} v_i v_i^{\intercal}) .
\end{split}
\end{equation}
where the second equality is due to Cauchy-Binet's formula, i.e.,  $k \geq d$ and $\abs{S} = k$.
\[ \mu(S) = \det(\sum_{i\in S} v_i v_i^\intercal) = \sum_{W \in \binom{S}{d}} \det(\sum_{i\in W} v_i v_i^{\intercal})\]
We can generalize this setting by considering $\nu: \binom{[n]}{\ell} \to \R_{\geq 0}$\footnote{For now think of $\ell$ as being equal to $d$, but more generally we are introducing parameter $\ell$ to unify the two cases of $k\leq d$ and $d\leq k$.} and $\mu: \binom{[n]}{k} \to \R_{\geq 0}$ by 
\begin{equation}\label{eq:low rank rayleigh}
    \mu(S) = \sum_{W \in \binom{S}{\ell}} \nu(W).
\end{equation}
We will assume that $\nu$ is strongly Rayleigh, and consequently satisfies the exchange inequality in \cref{lem:exchange inequality}.
Because \cref{lem:exchange inequality} only applies to local optima with $\nu(S)>0$, it is more convenient to work with a full-support distribution, i.e., $\nu(S) > 0$ for all $S\in \binom{[n]}{\ell}.$ Fortunately, we can approximate any strongly Rayleigh distribution $\nu$ with a full-support strongly Rayleigh distribution a.k.a. strictly real-stable distribution. In other words, for any $\epsilon > 0$, there exists a strongly Rayleigh distribution $\tilde{\nu}: \binom{[n]}{\ell} \to \R_{\geq 0}$ such that for all $S$, $\tilde{\nu} (S) > 0$ and $\abs{\tilde{\nu}(S) - \nu(S)}\leq \epsilon.$ Moreover, $\tilde{\nu}$ can be efficiently computed given $\nu$ (see the main theorem of \cite{Nuij1968ANO}, \cite[Proof of Proposition 2.2]{BH19} and \cite[page 7]{Branden21}). 
\begin{proposition} \label{prop:reduce to full support case}
Let $\nu: \binom{[n]}{\ell} \to \R_{\geq 0}$ be strongly Rayleigh. For any $\epsilon > 0$, there exists strongly Rayleigh $\tilde{\nu}:\binom{[n]}{\ell} \to \R_{> 0}$ such that $\abs{\nu(S)-\tilde{\nu}(S)}\leq \epsilon.$ for all $S$.
\end{proposition}
\begin{proof}[Proof of \cref{prop:reduce to full support case}]
By \cite{Nuij1968ANO}, for $i, j\in [n]$ and $s\in \R_{\geq 0},$ the following operator preserves strongly Rayleigh/real-stability of polynomials 
\[T_{i,j,s} g= g +s z_j \frac{\partial g}{\partial i} \]
for $ g\in \R[z_1,\dots, z_n].$

If $\nu(S) = 0$, for $\forall S$, then we can let $\tilde{\nu}(S) = \epsilon$ for $\forall S.$ W.l.o.g. assume $\nu(S) \neq 0$ for $S = \set*{1,\dots, \ell}.$
Let $g$ be the generating polynomial for $\nu$ and let
\begin{align*}
f = &\prod_{i\in [n], j\in [n]}T_{i,j,s} g = T_{1,1,s} \circ \dots \circ T_{1,n} \circ T_{2,1,s} \circ \dots T_{2,n,s} \circ \dots \circ T_{n,n,s} g
\end{align*}
then $f$ is strongly Rayleigh.
It is easy to see that for small enough $s,$ $f$ approximates $g.$ One practical choice for $s$ is $s = \epsilon (\sum_S \nu(S))^{-c}$ for some $c > 1;$ computing the partition function $\sum_S \nu(S)$ can be done efficiently and even in $\tilde{O}(1)$-parallel time for distributions of interest e.g. DPP. The map from $f$ to the multi-affine part 
$f^{MAP} $ of $f$ preserves real stability \cite{BBL09} (recall that for  $f(z_1,\dots, z_n) =\sum_{(\alpha_i)_{i=1}^n\in \N^n} c(\vec{\alpha}) \prod_{i=1}^n z_i^{\alpha_i},$ the multiaffine part of $f$ is 
$$f^{MAP}(z_1, \dots, z_n) = \sum_{(\alpha_i)_{i=1}^n \in {0,1}^n} c(\vec{\alpha}) \prod_{i=1}^n z_i^{\alpha_i}.$$ Now we only need to check that $f^{MAP}$ has positive coefficients.
Indeed, for $S=\set*{i_1, \dots, i_\ell}$ consider the coefficient of the monomial $z^S = \prod_{i\in S} z_i$ in $f$ and $f^{MAP}:$ it is a sum which includes the term
\[\prod_{j =1}^{\ell}(s z_{i_j} \frac{\partial} {\partial z_j}) \nu([\ell]) z^{[\ell]} = z^S s^k\nu([\ell]) > 0  \]
thus the coefficient of $z^S$ in $f^{MAP}$ is positive.
\end{proof}

\begin{remark} \label{remark:algorithm to find local opt}
We remark that a $O(1)$-approximate local optima of size $\ell$ w.r.t $\det(\cdot)$ can be found in time $O(n \cdot \poly(\ell)^4 ) $ using a combination of simple heuristics such as greedy and local search (also known as Fedorov exchange algorithm). The same algorithmic result holds more generally for all strongly Rayleigh distributions $ \nu:\binom{[n]}{\ell}\to \R_{\geq 0}$ \cite[see][for details]{anari2021sampling,ALOVV21}.
\end{remark}

In the remainder of the paper, we will consider the problem of maximizing the function $\mu:\binom{[n]}{k} \to \R_{\geq 0}$ under a matroid constraint where $\mu(S) = \sum_{W \in \binom{S}{\ell}} \nu(W)$ and $\nu: \binom{[n]}{\ell} \to \R_{> 0}.$ We assume $\nu$ is strongly Rayleigh, which implies that $\mu$ is also strongly Rayleigh.
\begin{proposition} \label{prop:sum strongly Rayleigh}
For $\nu: \binom{[n]}{\ell} \to \R_{\geq 0}$ being strongly Rayleigh and $ k \geq \ell,$ define   $\mu:\binom{[n]}{k} \to \R_{\geq 0} $ by \[\mu(S) = \sum_{W\in \binom{S}{\ell}} \nu(W)\] then  $\mu$ is also strongly Rayleigh.
\end{proposition}
\begin{proof}
Consider the elementary symmetric polynomial
\[e_{k-\ell}(z_1, \cdots, z_n) = \sum_{L \in \binom{[n]}{k-\ell}} \prod_{i\in L} z_i\]
then $e_{k-\ell}$ is real stable, i.e., has no roots in the upper half plane. Since the same is true for $g_{\nu}$, the product $g_{\nu}\cdot e_{k-\ell}$ also has no roots in the upper half plane \cite[see e.g.][Proposition 3.1 for a proof]{BBL09}. Consider the linear map $\varphi: \R[z_1,\cdots, z_n] \to \R[z_1, \cdots, z_n]$ that maps the monomial $z_1^{\alpha_1} \cdots z_n^{\alpha}$ to itself if $\alpha_i\leq 1$ for all $i,$ and to $0$ otherwise. This map preserves real-stability of polynomial \cite{Borcea_2009}, and $\varphi(g_{\nu} e_{k-\ell} ) = g_{\mu},$ thus $g_{\mu}$ is real stable, and $\mu$ is strongly Rayleigh. 
\end{proof}

\begin{remark}
This set-up of $\mu$ encompasses determinant maximization for both cases of $k\leq d$ and $k \geq d$. More concretely, for the former case, we set $\ell=k$, and for the latter case, we set $\ell=d$. We will explain this in more detail in \cref{thm:main}.

\end{remark}

For some constant $\zeta \geq 1 $ to be specified later, let $\phi(W) = (\zeta\ell)^{2 \abs{W \cap U} }.$ We define $\tilde{\mu}:[n]^k \to \R_{\geq 0}$ by:
\begin{equation}\label{eq:mu tilde}
    \tilde{\mu}(S) = \sum_{W \in \binom{S}{\ell}} \phi(W) \nu(W)
\end{equation}

This is the proper generalization of \cref{eq:modified determinant}. We observe the following simple fact.
\begin{fact} \label{fact:bound}
\[\mu(S) \leq \tilde{\mu}(S) \leq (\zeta \ell)^{2\ell} \mu(S) .\]
\end{fact}
\begin{lemma} \label{lem:smart exchange}
Let $V\subseteq [n]$ and let $U$ be a $\zeta$-approximate local optimum inside $V$ with respect to $\nu,$ for $\zeta =O(1).$ 
Then for any $e \in (V\cap S)\setminus U,$ there exists $f \in U$ s.t.
\[\tilde{\mu}(S) \leq \tilde{\mu} (S - e + f).\]
\end{lemma}
\begin{proof}
Consider $W \in\binom{S}{\ell}$ with $ e\in W.$ Using \cref{lem:exchange inequality} with $\ell' =\zeta\ell,$ we have 
\[\nu(W) \leq \ell' \sum_{f\in U} \nu(W - e+f) \]
Summing over all such $W$, we get
\begin{align*}
  &\sum_{W\in \binom{S}{\ell}: e\in W}\phi(W) \nu(W) \\
  \leq &\ell' \sum_{f \in U} \sum_{W: e\in W}\phi(W) \nu(W-e+f)\\
  \leq &(\zeta \ell)^2 \max_{f\in U}  \sum_{W\in\binom{S}{\ell}:e\in W} \phi(W) \nu(W-e+f) \\
  =&\sum_{W\in\binom{S}{\ell}:e\in W} \phi(W-e+f^*)\nu(W-e+f^*)     
\end{align*}
 with $f^*$ being the maximizer of the second line. Finally, adding $\sum_{W\in \binom{S-e}{\ell}}\phi(W) \nu (W)$ to both sides gives the desired inequality. 
\end{proof}
\subsection{The Algorithm} We have just shown how to exchange $e \in  (S \cap V)\setminus U$ for $f \in U$ while keeping $\tilde{\mu}$ non-decreasing. However, we still need to ensure that $ S-e+f$ is a proper set, i.e., ensure that $f \not\in S-e.$ To achieve this, we need a slightly more elaborate coreset construction.


\begin{definition}[Peeling coreset]
\label{def:threshold coreset}
Given $V \subseteq [n],$ and a number $ k_V \geq 1$, define the $(V,k_V, \zeta)$-peeling coreset $U $ as follows:

\begin{itemize}
\item Let $ U_0 = \emptyset .$ For $i =1, \cdots, k_V,$ let $ V_i:= V \setminus \bigcup_{j=0}^{i-1} U_j$, and let $U_i\subseteq V_i$ be a $\zeta$-approximate local optimum w.r.t. $\nu$ inside $V_i.$
    \item Let $ U = \bigcup_{i=1}^{k_V} U_i.$
\end{itemize}
Note that the $U_i$'s are disjoint and $ \abs{U} \leq k_V \ell.$

\end{definition}
\begin{lemma} \label{lem:non repetitive exchange}
The $(V, k_V)$-peeling coreset $U$ constructed in \cref{def:threshold coreset} is a value-preserving subset of $V$ with respect to $ \hat{\mu}: [n]^k \to \R_{\geq 0}$ defined as \[\hat{\mu}(S) = \textbf{1}[S\in \binom{[n]}{k}\land \abs{S \cap V}\leq k_V] \tilde{\mu}(S)\]
where $\tilde{\mu}$ is as defined in \cref{eq:mu tilde}.

\end{lemma}
\begin{proof}
Fix $ S\in \binom{[n]}{k}$ such that $\abs{S\cap V} \leq k_V$ 
and $e\in (S\cap V)\setminus U.$ Since $S$ has at most $ k_V-1$ elements inside\footnote{$\abs{S \cap U}\leq \abs{(S \cap V)\setminus e} \leq k_V-1.$} $U = \bigcup_{j=1}^{k_V} U_j$,  there exists some $j \in [k_V]$ such that $ S \cap U_j = \emptyset.$ Note that $e \in (S\cap V)\setminus U\subseteq (S\cap V_j)\setminus U_j, $. Thus, there exists $ f \in U_j$ such that $\tilde{\mu}(S) \leq \tilde{\mu}(S-e+f).$ Since $S \cap U_j = \emptyset,$ we are guaranteed that $f$ is not in $S-e.$
\end{proof}

\begin{lemma} \label{lem:low dimension} Let $\mu:\binom{[n]}{k} \to \R_{\geq 0}$ be strongly Rayleigh. For $ V\subseteq [n]$ and $\zeta \geq 1$, the $\zeta$-approximate local optimum $U$  w.r.t $\mu$ is a value-preserving subset of $V$ w.r.t $\hat{\mu}$ defined by
$\hat{\mu}(S) = (\zeta k)^{2\abs{U \cap S} } \mu(S).$
\end{lemma}
\begin{proof}
We use the fact that $\mu$ is strongly Rayleigh, and \cref{lem:exchange inequality}. For any $S \in \binom{[n]}{k}$ and $ e\in (S\cap V)\setminus U,$ there exists $j \in U \setminus S$ s.t.
\[\mu(S) \leq (\zeta k)^2 \mu(S-e+j) \]
Multiplying both sides by $ (\zeta  k)^{2\abs{U \cap S}}$ and using the fact that $\abs{(S-e+j) \cap U} = \abs{S \cap U} + 1$ we have 
\[\hat{\mu}(S) \leq \hat{\mu}(S-e+j)\]
\end{proof}
\begin{lemma}[Composability of value-preserving subsets] \label{lem:composable value}
Consider datasets $V^{(1)}, \cdots, V^{(m)}$ with $U_i$ being a value preserving subset of $V^{(i)}$ w.r.t $\tilde{\mu}$. Then $ U : = \bigcup_{i=1}^m U_i$ is a value-preserving subset of $V: = \bigcup_{i=1}^m V^{(i)}$ w.r.t $\tilde{\mu}.$ 
\end{lemma}
\begin{proof}
Consider $ S \subseteq [n]$ and $ e\in (S\cap V) \setminus U.$ Clearly, $ e\in  (S\cap V_i) \setminus U_i$ for some $i \in [m].$ Since $U_i$ is value-preserving w.r.t $\tilde{\mu},$ there exists $ j\in U_i\subseteq U$ s.t. $ \tilde{\mu}(S) \leq \tilde{\mu}(S-e+f).$
\end{proof}
\section{Composable Coresets for Partition and Laminar Matroids}
\label{sec:partition-laminar}
We construct composable coresets for determinant maximization under laminar matroid constraint. 
To build intuition, we first describe composable coresets for the simpler case of partition matroid. The idea is to build a peeling coreset of suitable size for each part of the partition which define the parition matroid.

As in \cref{sec:improve exchange}, given a matroid $\matroid$ with the set of bases $\B$, we consider the problem of maximizing $ \mu(S)$ (under matroid constraint) where $\mu(S) = \sum_{W\in \binom{S}{\ell}} \nu(W) $ and $ \nu :\binom{[n]}{\ell} \to \R_{\geq 0} $ is strongly Rayleigh. 
Let $\mu_{\matroid} $ be the restriction of $\mu$ to the set of bases of $\matroid$ i.e.  $\mu_{\matroid}(S) = \mathbb{1}[S \in \B(\matroid)] \mu(S).$
\begin{definition}\label{def:partition coreset}
Consider a partition matroid $\matroid = ([n], \mathcal{I})$ defined by the partition $P_1, \cdots, P_s$ of $[n]$ and $k_1, \cdots, k_s\in \N.$ Fix constant $\zeta \geq 1.$ For $V \subseteq [n],$ the composable coreset $U$ for $V$ w.r.t. $\mu_{\matroid}$ is constructed as follows: 
\begin{itemize}
    \item When $k > \ell:$ $U$ 
    is the union of $(V\cap P_i, k_i, \zeta)$-peeling coresets for each $i \in [s],$ thus $\abs{U} = k\ell.$ 
    \item When $k = \ell: U$  is the union over $i\in [s]$ of the $\zeta$-approximate local optimum w.r.t. $\nu$ in $V\cap P_i,$ thus $\abs{U} = s\ell = sk.$ 
\end{itemize} 
\end{definition}
\begin{lemma} \label{lem:partition matroid}The coreset constructed in
\cref{def:partition coreset} has an approximation factor of $(\zeta \ell)^{2\ell}.$
\end{lemma}
\begin{proof}
Note that in both cases, by \cref{lem:non repetitive exchange,lem:low dimension}, $U$ is the union of value-preserving subsets $U_i$ of $ V\cap P_i$ w.r.t $ \tilde{\mu}_{\matroid}(S) = \mathbf{1}[S \in \B]\sum_{W \in \binom{S}{\ell}} (\zeta\ell)^{2\abs{W \cap U}} \nu(W).$ Thus, by \cref{lem:composable value}, $ U$ is a value-preserving subset of $V$ w.r.t $ \tilde{\mu}_{\matroid}.$ \cref{fact:bound} and \cref{lem:value preserving to coreset} together imply that $ U $ is $\ell^{2\ell}$-composable coreset w.r.t. $\mu_{\matroid}.$ 
\end{proof}
We generalize the above construction to all laminar matroids.

\begin{definition} \label{def:laminar coreset}
Consider a laminar matroid over the ground set $[n]$ defined by a laminar family $\mathcal{F}$ and the associated integers $(k_S)_{S \in \mathcal{F}}.$ Fix constant $\zeta \geq 1.$ For $V\subseteq [n],$ the coreset for $V$ is constructed as follows:
\begin{enumerate}
    \item For each maximal set $F \in \mathcal{F},$ construct a coreset $U_F \subseteq V\cap F$ by:
     \begin{itemize}
         \item Let $ D_0 = \emptyset, V_0 = V \cap F.$ For $i=1, \cdots, k_S,$ let $U_i$ be the
         $\zeta$-approximate local optimal w.r.t. $\nu$ in $ V_i = V_{i-1} \setminus D_{i-1}.$ For $e \in U_i ,$ let $F^e\in \mathcal{F}$ be the maximal \emph{proper} subset of $ F $ containing $e$ or $\set{e}$ if no such $F^e$ exists. Let $D_i:=\bigcup_{e\in U_i} F^e.$ For each $e\in U_i$ with $F^e \neq \set*{e},$ recursively construct a coreset $U_{F^e}\subseteq V\cap F^e.$ 
         Observe that if no proper subset of $F$ is inside $\mathcal{F},$ then the coreset $U_F$ is precisely the $(V\cap F, k_F)$ peeling-coreset.
         \item The coreset $U_F$ for $F$ is the union of all $U_i$ and $U_{F^e}$ for $e\in U_i.$  
     \end{itemize}
     \item The coreset $U$ of $V$ is the union of all coresets $U_F$ for maximal sets $F \in \mathcal{F}$ and all elements $e\in V$ that do not belong to any set $F \in\mathcal{F}.$ 
\end{enumerate}
\end{definition}

\begin{theorem} \label{thm:main rayleigh laminar}
Consider $ \nu: \binom{[n]}{\ell} \to\R_{\geq 0}$ that is strongly Rayleigh and $ \mu: \binom{[n]}{k} \to \R_{\geq 0}$ s.t.
\[\mu(S) = \sum_{W\in \binom{S}{\ell}} \nu(W).\]
Consider laminar matroid constraint $\mathcal{M}$ of rank $k$ defined by non-redundant family $\mathcal{F}$ with cover number $r$ i.e.  $$r: = \max_{e\in [n]} \abs{ \set*{F \in\mathcal{F}: e\in F} }.$$ \cref{def:laminar coreset} gives a $(\zeta\ell)^{2\ell}$-composable coreset w.r.t $\mu$ under matroid constraint $\matroid$ of size at most $(\zeta k\ell)^r\leq (\zeta k\ell)^k.$
\end{theorem}
\begin{proof}[Proof of the approximation factor]

Let \[\tilde{\mu}(S) = \sum_{W\in \binom{S}{\ell}} (\zeta\ell)^{2 \abs{W \cap U}} \nu(W) .\]
$U$ is a value-preserving subset of $V$ w.r.t the restriction $\tilde{\mu}_{\matroid}$ of $\tilde{\mu}$ to the set of bases of the laminar matroid i.e. $\tilde{\mu}_{\matroid}(S) = \mathbb{1}[S \in \B(\matroid)] \tilde{\mu}(S).$ This combined with \cref{lem:value preserving to coreset} immediately imply that $ U$ is a $\ell^{2\ell}$-composable coreset w.r.t $\mu_{\matroid}.$ 

We only need to show that $U_F$ is value preserving for each $F\in \mathcal{F}.$ Fix $S\in \B$ and $ h \in (S \cap V\cap F) \setminus U_F.$ We claim that there exists $f \in U_F$ s.t. $S-h +f \in B$ and $\hat{\mu}(S) \leq \hat{\mu}(S-h+f).$ 

We prove this by induction on $F.$
For the base case when $F$ has no proper subset inside $\mathcal{F},$ then $U_F$ is the $(V\cap F, k_F)$ peeling-coreset, and the claim follows from \cref{lem:non repetitive exchange}. 
If $h \in D_i$ for some $i,$ then $h$ must be contained in a proper subset $F^e\in \mathcal{F}$ of $F$ where $e\in U_i$ and $F^e \neq \set*{e}$\footnote{if $F^e = \set*{e}$ then $h =e \in U_i \subseteq U_F $, a contradiction} and we can use the induction hypothesis. Now, assume $h\not\in D_i$ for $\forall i\in [k_F].$ In particular, this means $h \in V_{k_F}\subseteq \cdots\subseteq V_1 = V \cap F$ and $D_i$ and $U_i$ are non-empty for all $i\in [k_F]. $ Note that since $D_i$'s are disjoint, and $S$ contains at most $k_F -1$ elements inside $F,$ $S \cap D_i = \emptyset$ for some $i.$ In particular, $S \cap U_i = \emptyset$ and $h \in (V_i\cap S)\setminus U_i,$ so \cref{lem:smart exchange} implies that there exists $ f \in U_i$ s.t. $ \hat{\mu}(S) \leq \hat{\mu}(S-h +f).$ Replacing $h$ with $f$ only affects the constraints for sets $F'\in \mathcal{F}$ containing $f.$ Consider such a set $F'.$ $ F'$ must be contained inside $D_i$ by the definition of $D_i,$ thus $S \cap F' = \emptyset,$ and $\abs{(S-h+f) \cap F'} \leq 1\leq k_{F'}.$ We just verify that $S-h+f$ is also a base of the laminar matroid, thus \[\hat{\mu}(S-h+f) = \tilde{\mu}(S-h+f) \geq \tilde{\mu}(S) = \hat{\mu}(S).\]
\end{proof}
\begin{proof}[Proof of upper bound on the size of the coreset]
For a set $H$ let $r_H: = \max_{e\in H} \abs{\set*{F\in \mathcal{F}: F \subseteq H \land e\in F} }.$ We show that $ \abs{U_F} \leq (k_F \ell)^{r_F}$ for each $F\in \mathcal{F}$ by induction on $r_F.$ For the base case $r_F = 1,$ we have $\abs{U_F} =k_F \ell,$ by \cref{def:threshold coreset}. Fix $F\in \mathcal{F}$ with $r_F \geq 2$ and suppose the induction hypothesis holds for $r< r_F.$ Using the definition of $U_F,$ we can bound 
\begin{align*}
  \abs{U_F} &\leq \sum_{i=1}^{k_F} \abs{U_i} + \sum_{e\in \bigcup_{i=1}^{k_F} U_i} \abs{U_{F^e}} \\
  &\leq_{(1)} k_F d + (k_F \ell) (d\max_{F'\subseteq F: F ' \in\mathcal{F} } k_{F'} )^{r_F-1} \\
  &\leq_{(2)} (k_F \ell)^{r_F}  
\end{align*}
where in (1) we use the fact that $ r_{F^e} < r_F$ since $ F^e$ is a proper subset of $F,$ and in (2) we use \[(\max_{F'\subseteq F: F ' \in\mathcal{F} } k_{F'})^{r_F-1} + 1 \leq (k_F -1)^{r_F-1} + 1\leq k_F^{r_F-1}.\]
Thus the induction hypothesis holds for all $r.$

Suppose the maximal set(s) in $\mathcal{F}$ are $ F_1, \cdots, F_t,$ and let $ R:= [n] \setminus \bigcup_{i=1}^t F_t.$ Then the rank of the laminar matroid is $k =\abs{R}+ \sum_{i=1}^t k_{F_t} ,$ and
\[\abs{U} = \abs{R \cup \sum_{i=1}^t U_{F_t}} \leq \abs{R} + \sum_{i=1}^t (k_{F_t} \ell)^{r} \leq (k\ell)^r.\]
\end{proof}
\begin{remark}
For any laminar family $\mathcal{F}$ of rank $k,$ we can construct a $d^{O(d)}$-coreset of size $ \abs{\mathcal{F}} d k$ by taking the union of all value-preserving subsets of $ V\cap (F \setminus \bigcup_{F'\subseteq F, F'\in \mathcal{F}} F')  .$ However, the size of the coreset might be as bad as linear in $n.$ Indeed, consider the laminar family defined by: $F_i = \set*{2i+1, 2i+2} , k_{F_i} = 1$ for $\forall i\in [n/2]$ and $ F_0 = [n],$ $k_{F_0} = k $ then \cref{def:laminar coreset} gives a coreset of size\footnote{We can improve the bound to $ kd^2$ by a more careful analysis.} $\leq k^2 d^2$ whereas the naive construction gives a coreset of size $\geq (n/2) d.$
\end{remark}

We immediately obtain the following corollary about determinant maximization under matroid constraints.
\begin{theorem} \label{thm:main}
For the determinant maximization matroid constraints with input vectors $v_1, \cdots, v_n \in 
\R^d,$ we obtain the following results:
\begin{enumerate}
    \item Partition matroid defined by partition $P_1, \cdots, P_s$ of $[n],$ \cref{def:partition coreset} gives:
    \begin{itemize}
    \item For $k \leq d:$ $k^{2k}$-composable coreset of size $O(sk).$ 
    \item For $k \geq d:$ $ d^{2d}$-composable coreset of size $O(kd).$
\end{itemize}
    \item Laminar matroid:
    \begin{itemize}
    \item For $k \leq d:$ $k^{2k}$-composable coreset of size $O(k^{2k}).$
    \item For $k \geq d:$ $ d^{2d}$-composable coreset of size $O((kd)^k).$
\end{itemize}
\end{enumerate}

\end{theorem}
\begin{proof}
We show how to adapt the setting of $\nu:\binom{[n]}{\ell}\to\R_{\geq 0}$ and $\mu:\binom{[n]}{k} \to \R_{\geq 0}$ where $\mu(S) = \sum_{W \in \binom{S}{\ell}} \nu(W)$ with $\nu$ being strongly Rayleigh to the determinant maximization setting.
\begin{itemize}
    \item For $k \leq d:$ we let $\ell = k$ and $\mu(S) =\nu(S) = \det(\sum_{i\in S} v_i v_i^\intercal)$ for $|S| =k.$ By replacing $\ell$ with $k$ we get the stated result. 
    \item For $ k \geq d:$ we let $\ell = d,$ $\nu(W) = \det(\sum_{i\in W} v_i v_i^\intercal)$ for $\abs{W} =\ell$ and $\mu(S) =\nu(S) = \det(\sum_{i\in S} v_i v_i^\intercal)$ for $|S| =k.$
\end{itemize}
\end{proof}

Recall that $O(1)$-approximate local optima can be found in time $O(n\poly(k))$ (see \cref{remark:algorithm to find local opt}). Thus, our coreset construction is highly efficient: it takes time $O(n \poly(k))$ for the case of partition matroid constraint. As a corollary, we obtain a quasilinear algorithm for MAP-inference for DPP under partition matroid constraint.
\begin{lemma} \label{lem:linear time algo partition matroid}
Consider a partition matroid $\matroid = ([n], \mathcal{I})$ of rank $k$ defined by the partition $P_1, \cdots, P_s$ of $[n]$ and $k_1, \cdots, k_s\in \N.$ Given input vectors $v_1, \cdots, v_n\in \R^d,$ there exists a $O(n \poly(k))$ algorithm that outputs a $\min (k^{O(k)} , d^{O(d)})$-approximation for the determinant maximization under partition matroid constraint $\matroid$.  

\end{lemma}
\begin{proof}
W.l.o.g. we can assume $k_i \geq 1$ for all $i.$ We construct coreset $U$ as in \cref{thm:main}.
Note that since $k_1+\cdots + k_s = k$ and $k_i \geq 1$, we have that $s \leq k$ thus the size of $U$ is $O(k^2)$ for both cases $k \leq d$ and $k\geq d.$ We can restrict the ground set to $U$ and use the existing efficient algorithms \cite{BLPST22} to get a $\min\set*{k^{O(k)}, d^{O(d)}}$-approximation for constrained determinant maximization with input vectors from $U$, which is also a $\min\set*{k^{O(k)}, d^{O(d)}}$-approximation for the original constrained determinant maximization problem.   
\end{proof}
\section{Other Experimental Design Problems} \label{sec:design}
In this section, we generalize our composable coreset construction to other experimental design problems such as $A$-design and $E$-design.

The main idea is to replace the local optimum in the coreset construction with an $\alpha$-spectral spanner (see \cref{def:spectral spanner}). 
By replacing the local optimum with a spectral spanner, we can ensure that the coreset contains a high-valued feasible fractional $x$ in the convex hull $ P(\matroid)\subseteq [0,1]^n$ of the matroid polytope of $\matroid,$ which can be rounded to an integral solution for uniform matroid constraint and certain class of laminar matroid constraint.

\begin{definition}[{\cite{indyk2020composable}}] \label{def:spectral spanner}
For a set of vectors $V\subseteq \R^d$, a subset $U\subseteq V$ is a $\alpha$-spectral spanner of $V$ iff for any $ v\in V,$ there exists a distribution $\mu_v$ of vectors in $U$ s.t.
\[v v^\intercal \preceq \alpha \E_{u \sim \mu_v} {u u^\intercal}\]
\end{definition}
\begin{theorem}[{\cite[Proposition 4.2,Lemma 4.6]{indyk2020composable}}]\label{thm:spectral spanner}
Given $V\subseteq \R^d,$
there exists an efficient algorithm that constructs $\tilde{O}(d)$-spectral spanner of size $\tilde{O}(d).$
\end{theorem}

Recall that the goal of experimental design problem is to select a set $S$\footnote{$S$ might need to satisfy additional constraints such as $S$ is a basis of a given matroid} that maximizes $f(\sum_{xxi\in S} v_i v_i^\intercal)$ for some objective function $f$. The most popular and well-studied objective functions include:
\begin{itemize}
    \item D(eterminant)-design: $f(A) = \det(A)^{1/d}.$
    \item A(verage)-design: $f(A) = -\Tr(A^{-1})/d$
    \item E(igen)-design: $f(A) = -\norm{A^{-1}}_2$
    \item T(race)-design: $f(A) = d/\Tr(A)$
\end{itemize}
Each of the above objective functions satisfies the properties of a \emph{regular} function  (see \cref{def:regular objective function}). \cite{roundingExperimentalDesign} shows that under uniform matroid i.e. cardinality constraint, any fractional feasible solution of a regular function can be rounded into an integral solution while incurring only $O(1)$ loss in the objective function.
For laminar matroids, \cite{lau2021local} shows the same results for $D$-design and $A$-design when $k_F \geq  C d$ for $\forall F \in \mathcal{F}$ and for some absolute constant $C.$ For general matroid and $f(\cdot) = \det(\cdot),$ \cite{Madan2020MaximizingDU} shows that a fractional feasible solution can be rounded into an integral solution while suffering a $ d^{O(d)}$ loss \emph{in expectation}.

\begin{definition} \label{def:regular objective function}
A function $f: \Sd \to \R$ is regular if it satisfies the following properties
\begin{itemize}
    \item Monotonicity: for any $A, B \in \Sd,$ if $f (A) \leq f(B)$ for $ A \preceq B.$
    \item Concavity:  for $A, B \in \Sd$ and $t\in[0,1]$, we have $f (tA + (1-t) B) \geq t f(A) + (1-t) f(B) .$
    
    In particular, this implies the existence of an efficient algorithm that solves the continuous relaxation 
    \[\max_{s_1,\cdots,s_n} f(\sum_{i\in 1}^n s_i v_i v_i^\intercal) \text{ s.t. } s_i \in [0,1] \text{ and } \sum_{i=1}^n  s_i \leq k.\]
    \item Reciprocal linearity: for any $A \in \Sd$ and $ t\in (0,1),$ $f(tA) = t^{-1} f(A).$ 
\end{itemize}
\end{definition}

\begin{theorem}[Rounding for experimental design,
{\cite{Madan2020MaximizingDU}}] \label{thm:rounding}
Consider the experimental design problem with objective function $f(\cdot)$ and input vectors $v_1, \cdots, v_n \in \R^d$ under matroid constraint $\matroid$ of rank $k.$
For any fractional $x\in P(\matroid)\subseteq [0,1]^n,$ there exists $z \in \mathcal{B}(\matroid)\subseteq \set*{0,1}^n$ s.t.
\begin{itemize}
    \item When $f(A) = \det(A):$
    \[\min\set*{d^{O(d)}, 2^{O(k)}} f(\sum_{i=1}^n z_i v_i v_i^\intercal) \geq f(\sum_{i=1}^n x_i v_i v_i^\intercal) \]
    The factor $d^{O(d)}$ can be improved to $ 2^{O(d)}$ when $\matroid$ is a partition matroid.
    \item When $k \geq d ,$ $\matroid$ is the uniform matroid and $f$ is regular:
    \[O(1) f(\sum_{i=1}^n z_i v_i v_i^\intercal) \geq f(\sum_{i=1}^n x_i v_i v_i^\intercal) \]
    \item When $ k \geq d,$ $\matroid$ is a laminar matroid defined by the laminar family $\mathcal{F}$ and $(k_F)_{F\in \mathcal{F}}$ with $ k_F \geq C d$ for $\forall F \in \mathcal{F}$ for some large absolute constant $C$, and $f(A) = -Tr(A^{-1})/d$ :
    \[O(1) f(\sum_{i=1}^n z_i v_i v_i^\intercal) \geq f(\sum_{i=1}^n x_i v_i v_i^\intercal)\]
\end{itemize}

\end{theorem}
We show $\tilde{O}(d)$-composable coreset of size $\tilde{O}(dk)$ for experimental design problems in the without repetition setting.
\begin{theorem}\label{thm:design}
Given input vectors $v_1, \cdots, v_n\in \R^d,$ $V \subseteq \set*{v_1, \cdots, v_n}$ and a number $k_V\geq 1,$ the $(V,k_V)$-spectral peeling coreset $U$ is defined by the same procedure as in \cref{def:threshold coreset}, but replacing the local optimal $U_i$ by a $O(d)$-spectral spanner $U_i$ of $V_i$ (see \cref{def:spectral spanner,thm:spectral spanner}). Then $\abs{U} \leq \tilde{O}(k_V d).$ For any $ S $ with $\abs{S\cap V} \leq k_V,$ there exists a distribution $ \mu_v $ for $v \in S\cap V $ with disjoint supports s.t. $ \supp(\mu_v) \subseteq U$ and for any regular objective function $f:$
\[f(\sum_{v\in S } v v^\intercal ) \leq  f(\sum_{v\in S} \E_{u\sim \mu_v} { d u u^\intercal })\]
Consequently, for $k_V = k,$ the set $ U$ serves as a $\tilde{O}(d)$-composable coreset for the experimental design problem under cardinality constraint $k$ w.r.t $f.$
\end{theorem}
\begin{proof} 
Consider the set $  (S \cap V) \setminus U.$ Since $\abs{S\cap V} \leq k_V,$ there is an injective map $ \pi: (S\cap V) \setminus U \to [k_V]$ s.t. $S \cap U_{\pi(v)} = \emptyset$ for each $v \in (S\cap V) \setminus U.$ 
For each $v\in (S\cap V)\setminus U$, since $v\in V_{\pi(v)}$, we can use the fact that $U_{\pi(v)}$ is a spectral spanner of $V_i$ to deduce that there exists $\mu_v$ supported on $U_{\pi(v)}$ where 
\[v v^\intercal \preceq d\E_{u \sim \mu_v}{u u^\intercal}.\]
Note that $\mu_v$ are disjoint by injectivity of $\pi.$
The claim then follows from the monotonicity of $f.$

For sets $V_1, \cdots, V_m$ let $U'_i$ be the $(V_i,k)$-peeling coreset for $V_i.$ Let $V':= \bigcup_{i=1}^m V_i$ and $U':= \bigcup U'_i.$

Let $S\in \binom{V'}{k}$ be a subset that maximizes $ f(\sum_{v\in S} v v^\intercal)$. Using the above argument, we obtain that $U'$
 contains a fractional solution $s\in [0,1]^{U'}$ s.t. $ \sum s_i = k$ and
\[f(\sum_{v\in S} v v^\intercal ) \leq d f(\sum_{i\in U'} s_i  v_i v_i^\intercal) \leq O(d) f (\sum_{u\in \tilde{S}} u u^\intercal ) \]
for some $\tilde{S} \in \binom{U'}{k},$ where the second inequality follows from \cref{thm:rounding}.
\end{proof}

Using similar construction and proof technique, we obtain $\tilde{O}(d)$-composable coreset of size $\tilde{O} (dk)$ and $\tilde{O}((dk)^k)$ respectively for $A$-design under certain laminar and partition matroid constraint $\matroid$ where $k_F \geq Cd$ for $\forall F\in \mathcal{F}.$

\section{Lower Bounds}\label{sec:lower-bound}

In this section, we show that the coreset we constructed essentially attains the best possible size and approximation factor. 
We first show that for determinant maximization in $\R^d$ when $k \leq d$ under partition matroid constraint, our coreset size is optimal.
\begin{lemma} \label{lem:size lb in low dimension}
Suppose $ k \leq d.$ Consider a partition matroid $\matroid =([n], \mathcal{I})$ defined by a partition $P_1 \cup \cdots\cup P_s = [n]$ and constraint $k_1, \cdots ,k_s.$ Let $k : = \rank(\matroid)= \sum_{i=1}^s k_i.$  Any $\alpha$-composable coreset for the determinant maximization problem under partition matroid constraint $\matroid$ with size $ t< sk$ must incur an arbitrarily large approximation factor. 
\end{lemma}
\begin{proof}
Consider $n$ vectors $v_1, \cdots, v_n\in \R^d$ to be chosen later. For set $Q \subseteq [n],$ let $$\OPT(Q):=\max_{S\subseteq Q, |S| = k} \det(\sum_{i\in S} v_i v_i^\intercal).$$
Consider a partitioning of $[n]$ into two sets $Q, Q'$ such that for each $i,$ $\set{v_j: j\in Q\cap P_i}=\set*{e_1, \cdots, e_d}$ where $e_1, \cdots, e_d$ is the standard basis for $ \R^d.$ We need to show that for any subset $U\subseteq S$ of size $ t < sk,$ we can choose the vectors in $Q'$ s.t. $\OPT(Q \cup Q') \gg \OPT(U \cup Q').$ Indeed, fix one such subset $U$ where $|U| \leq sk-1$. Since $\sum_{i\in [s]} |U \cap P_i| \leq |U |\leq sk-1,$ there must exist $i\in [s]$ s.t. $\abs{U\cap P_i}\leq k-1. $ W.l.o.g., we can assume that $\set{v_j: j\in U\cap P_1} \subseteq \set*{e_1, \cdots, e_{k-1}}.$ Choose $Q'$ s.t. $Q' \cap P_1 = \emptyset, $ and \[ \set{v_j: j\in Q' \cap P_i} = \set*{Me_{\sum_{j=1}^{i-1} k_j }, \cdots, M e_{\sum_{j=1}^i k_j -1 }}\] for some arbitrarily large $M > 0.$  Consider $S\subseteq Q\cup Q'$ s.t. $\set{v_j :j\in S\cap P_1} = \set*{e_1, \cdots, e_{k_1-1}, e_k}$ and $\set{v_j: j\in S\cap P_i = Q'\cap P_i}$ for $i=2, \cdots, s$, then $S\in \matroid$ and $ \mu(S):=\det(\sum_{i\in S} v_i v_i^\intercal) = M^{2\sum_{j=2}^s k_j},$ thus $\OPT(Q \cup Q') \geq M^{2\sum_{j=2}^s k_j}.$ On the other hand, for any $S'\in \binom{U \cup Q'}{k},$ either:
\begin{itemize}
\item $S'\cap P_i \subseteq Q'\forall i\geq 2:$ in this case $\mu(S') = 0$\footnote{Even when we replace $\mu$ by a full-support $\tilde{\mu}$ that approximates $\mu$ within distance $\epsilon$, we will have $\tilde{\mu}(S) < \epsilon<\OPT(V\cup V')/M^2$ if we choose $\epsilon$ small enough} because all the vectors in $S'$ are contained in the $(k-1)$-dimensional subspace spanned by $e_1, \cdots, e_{k-1}.$
    \item $S'\cap P_i \not\subseteq Q'$ for some $i\geq 2:$ in this case $\mu(S') \leq M^{2 (\sum_{j=2}^s k_j-1)}$ since there are at most $\sum_{j=2}^s k_j -1 $ vectors in $S'$ that are from $V'$ and thus have norm $M,$ while the remaining vectors are from $V$ and have norm $1.$
\end{itemize}
In either case, we have $\OPT(U\cup Q') \leq M^{2(\sum_{j=2}^s k_j-1)}\leq \OPT(Q \cup Q')/M^2 ,$ thus $\OPT(Q\cup Q')$ can be arbitrarily large compared to $\OPT(U \cup Q').$
\end{proof}
For $k \geq d,$ using similar arguments, we can show that any $\alpha$-coreset for determinant maximization under partition constraint with finite approximation factor $\alpha$ must have size $ t \geq k+d(d-1).$
\begin{lemma} \label{lem:size lb in high dimension}
Suppose $k \geq d.$ Consider the partition matroid $\matroid =([n], \mathcal{I})$ of rank $k$ defined by a partition $P_1,\cdots, P_k$ 
and constraint $k_1= \cdots = k_k =1 .$ Any composable coreset for the determinant maximization problem under partition matroid constraint $\matroid$ with size $t< k + d(d-1)$ must incur an arbitrarily large approximation factor. 
\end{lemma}
\begin{proof}
The construction is similar to the proof of \cref{lem:size lb in low dimension}. Consider $n$ vectors $v_1, \cdots, v_n\in \R^d$ to be chosen later. For set $Q \subseteq [n],$ let $\OPT(Q):=\max_{S\subseteq Q, |S| = k} \det(\sum_{i\in S} v_i v_i^\intercal).$

Choose set $Q\subseteq [n]$ and $\set{v_j:j \in Q}$ s.t. \[\set{v_j: j\in Q\cap P_i} = \set*{M_i e_1, \cdots, M_i e_d}\] 
with $M_1 \geq  M_2 \geq M_d \gg M_{d+1} \cdots \geq M_k$ to be chosen later. Let $U$ be a coreset for $Q$ with finite approximation factor. Clearly, $\abs{Q \cap P_i} \geq 1.$ 
We will show that $ \abs{U\cap P_i} = d $ for $i=1, \cdots, d,$ and thus conclude that $ \abs{U} \geq (k-d) + d^2 = k+d(d-1).$

For the base case of $i=0,$ the claim holds trivially. 
Suppose that the claim holds for $ i-1 $ with $i\geq 1$. Then we show that it holds for $i.$ We assume for contradiction that $ \abs{U\cap P_i} \leq d-1.$ W.l.o.g., assume $ \set{v_j: j\in U\cap P_i} \subseteq {M_ie_1, \cdots, M_ie_{i-1}, M_i e_{i+1}, \cdots, M_i e_{d}}.$
Indeed, define $Q'$ such that $\set{v_j: j\in Q'\cap P_t} = \set*{M e_t} $ for $t \in \set*{1, \cdots, i-1, i+1, \cdots, d}$ and $Q'\cap P_t = \emptyset$ otherwise. By choosing $ M\gg M_1$ and $M_i\gg M_d \gg M_{d+1},$ we can ensure that the optimal instance in $ Q\cup Q'$ must contain $d-1$ vectors in $Q'$ and $M_i e_i \in Q \cap P_i,$ thus \[\OPT(Q\cup Q') \geq  (M^{d-1} M_i)^2.\]
On the other hand, for any $S \subseteq U \cup Q',$ either
\begin{itemize}
    \item $\abs{S \cap Q'} \leq d-1:$ in this case \[\mu(S) \leq \binom{k}{d} (M^{d-2} M_1^2)^2 \] since any $W \in \binom{S}{d}$ must contain at most $d-2$ vectors of norm $M$ from $V',$ and the remaining vectors have norm at most $M_1.$
    \item $\abs{S \cap Q'} = d-1:$ since $\set{v_j: U\cap P_i }$ is in the span of $\set{v_j:S \cap Q'},$ any $W\in \binom{S}{d}$ with $\det(\sum_{i\in W} v_i v_i^\intercal) \neq 0$ must satisfy that $V_W :=\set{v_j:j\in W}$ consist of at most $d-1$ vectors of norm $M$ from $Q',$ and the remaining vectors must have norm at most $M_{d+1},$ thus \[\mu(S) \leq \binom{k}{d} (M^{d-1} M_{d+1})^2 .\]
\end{itemize}
In either case, $\OPT(U\cup Q')$ can be arbitrarily smaller than $ \OPT(Q\cup Q').$ 
\end{proof}

Finally, we show that for $k\geq d,$ the approximation factor of $d^{O(d)}$ is the best possible even under no constraints. For $k\leq d,$ \cite{indyk2020composable} shows that approximation factor of $k^{O(k)}$ is optimal.

The following construction is from \cite[section 7.1]{indyk2020composable}. We include this for completeness.
\begin{definition}[Hard input for composable coreset]\label{def:hard}
Let $ \beta = o(d/\log^2 d),$ $m = d/\log d$ so that $d^{d/m} = O(1).$ Consider $ G\subseteq \R^{m+1}$ of $ d^{\beta+2}$ vectors s.t. for every two vectors $p, q \in G,$ we have $\langle p ,  q\rangle \leq O(\frac{\sqrt{\beta} \log d}{\sqrt{d}}).$  

For $i =1, \cdots, d-m$, construct $X_i$ as follows: pick a random index $\pi_i\in [n].$ Embed $ G$ into the subspace spanned by $\set*{e_1, \cdots, e_m, e_{m+i}}$ s.t. the $\pi(i)^{th}$ vectors in $G$ is mapped into $e_{m+i}.$

Choose a random rotation matrix $Q,$ and return $Q X_1, \cdots,$ $Q X_{d-m}$ and $ QY_1, \cdots, QY_m$ with $ Y_i = \set*{M e_i}$ for a large enough $M.$
\end{definition}
\begin{theorem}\label{thm:lowerbound}
For $d\leq k \leq d^{\gamma}$ and $\gamma'$ s.t. $ \gamma \gamma' = o(d/\log^2 d),$ any composable coreset of size $k^{\gamma'}$ must incur an approximation of $(\frac{d}{\gamma \gamma'})^{d(1-o(1))}.$
For example, the theorem applies when $ \gamma, \gamma' $ are constant, i.e. $ d\leq k \leq \poly(d)$ and the coreset has size $\poly(k).$
\end{theorem}
\begin{proof}
For a set $V$ of vectors, let $ V^{\times t}$ be the set where each vector in $V$ is duplicated $t$ times. Let $ \beta = \gamma \gamma'.$
We use the construction in \cref{def:hard} where every vector is duplicated $t=k/d$ times. Let $QX_1^{\times t}, \cdots, QX_{d-m}^{\times t},$ $QY_1^{\times t}, \cdots, QY_m^{\times t}$ be the input sets. Let $S$ be s.t. $$V_S:=\set{v_j: j\in S} = \set*{Me_1, \cdots, M e _m, e_{m+1}, \cdots, e_{d}}$$ then $V_S^{\times t}$ has value $\mu(S) \geq (k/d)^d (M^m)^2. $

On the other hand, let $c(Q X_i^{\times t})$ be an arbitrary coreset of size $k^{\gamma'} \leq d^{\beta}$ for $  Q X_i^{\times t}.$ 

As observed in \cite[Lemma 7.2]{indyk2020composable}, the probability that $C_i:= c(QX_i^{\times t})$ contains $Q e_{m+i}$ is bounded by $$\abs{c(QX_i^{\times t})}/\abs{QX_i^{\times t})}\leq 1/d^2.$$ Thus, with probability $\geq 1-1/d,$ we have $ Qe_{m+i} \not \in c(QX_i^{\times t})$ for all $i \in [d-m].$ Assume that this happens. Then for any $u\in C_i$
\[\langle \sum_{i=m+1}^d e_i e_i^\intercal, u u^\intercal \rangle \leq O(\frac{\beta \log^2 d }{d})\]
thus for any $ u_1, \cdots, u_{m}$ in  \[ \mathcal{C}:=\bigcup_{i=1}^m c(QY_i^{\times t}) \cup \bigcup_{i=1}^{d-m} c(QX_i^{\times t}) ,\]
\begin{align*}
  &\det(\sum_{i=m+1}^d (M e_i)(Me_i)^\intercal +\sum_{i=1}^m u_i u_i^\intercal) \\
  \leq &M^{2m} (\max \langle \sum_{i=m+1}^d e_i e_i^\intercal, u u^\intercal \rangle)^{d-m} \\
  \leq  &M^{2m} \left(\frac{ O(\sqrt{\beta}) \log d}{d}\right)^{2(d-m)}.  
\end{align*}

Hence, with probability at least $1-1/d,$ any size-$d$ subset $W$ in $ \mathcal{C}$ has \[\det(\sum_{v\in W} v v^{\intercal}) \leq   M^{2m} \left(\frac{ O(\sqrt{\beta}) \log d}{d}\right)^{2(d-m)},\] thus by Cauchy Binet, any size-$k$ subset $S$ in $ \mathcal{C}$ has
\[\mu(S) \leq \binom{k}{d}  M^{2m} \left(\frac{ O(\sqrt{\beta}) \log d}{d}\right)^{2(d-m)}.\]
Thus the approximation factor is at least \[\frac{1}{e^d}(d/ (O(\sqrt{\beta}) \log d )^2)^{d-m} \] with $m = o(d).$

\end{proof}
\bibliographystyle{plainnat}  
\bibliography{refs}

\end{document}